\documentclass{easychair}

\usepackage{doc}

\usepackage{amsthm}
\newtheorem{theorem}{Theorem}
\newtheorem{definition}[theorem]{Definition}
\newtheorem{lemma}[theorem]{Lemma}
\newtheorem{proposition}[theorem]{Proposition}
\newtheorem{corollary}[theorem]{Corollary}

\usepackage[utf8]{inputenc}
\usepackage{mathtools}
\usepackage{xparse}

\DeclarePairedDelimiterX{\set}[1]{\{}{\}}{\setargs{#1}}
\NewDocumentCommand{\setargs}{>{\SplitArgument{1}{;}}m}
{\setargsaux#1}
\NewDocumentCommand{\setargsaux}{mm}
{\IfNoValueTF{#2}{#1} {#1\nonscript\:\delimsize\vert\allowbreak\nonscript\:\mathopen{}#2}}

\usepackage{ stmaryrd }

\DeclarePairedDelimiterX{\mset}[1]{\llparenthesis}{\rrparenthesis}{\msetargs{#1}}
\NewDocumentCommand{\msetargs}{>{\SplitArgument{1}{;}}m}
{\msetargsaux#1}
\NewDocumentCommand{\msetargsaux}{mm}
{\IfNoValueTF{#2}{#1} {#1\nonscript\:\delimsize\vert\allowbreak\nonscript\:\mathopen{}#2}}

\usepackage{csquotes}

\usepackage{enumitem}

\usepackage{complexity}

\usepackage{amsfonts}

\usepackage[english]{babel}

\theoremstyle{plain}

\usepackage{amssymb}

\newcommand{\N}{\mathbb{N}}

\newcommand{\lcard}{\textsf{LIA}^{card}}

\usepackage{listings}

\usepackage{booktabs}

\def\i{n}

\title{Combinatory Array Logic with Sums}

\author{
Rodrigo Raya\inst{1}\thanks{Research supported by the Swiss NSF Project P500PT\_222338}
}

\institute{
  Max-Planck Institute for Software Systems,
  Kaiserslautern, Germany\\
  \email{rraya@mpi-sws.org}
 }

\authorrunning{R. Raya}

\titlerunning{Combinatory Array Logic with Sums}

\begin{document}

\maketitle

\begin{abstract}
We prove an $\NP$ upper bound on a theory of integer-indexed integer-valued arrays that extends combinatory array logic with the ability to express sums of elements.
The decision procedure that we give is based on observations obtained from our analysis of the theory of power structures. 
\end{abstract}


\section{Introduction}
\label{section:intro}
Many applications of computer science to operations research and software engineering require some form of constraint solving technology. We focus in the satisfiability modulo theories (SMT) framework which was intensively developed in the first decade of the century, leveraging progress in the architecture of propositional satisfiability solvers \cite{moskewicz_chaff_2001,  nieuwenhuis_solving_2006, de_moura_z3_2008}. 

SMT addresses the satisfiability problem of fragments of first-order theories that are quantifier-free or have a small number of quantifier alternations. In fortunate occasions, this restriction makes the satisfiability problem $\NP$-complete. In such cases, it is possible to reduce the satisfiability problem of the fragments to the satisfiability problem of propositional logic in polynomial time. Some theories supported using such reduction include real numbers, integers, lists, arrays, bit vectors, and strings \cite{bradley_calculus_2007, kroening_decision_2016}.

This work analyses the structure of a well-known fragment of the quantifer-free theory of arrays. In the SMT framework, arrays are conceived as indexed homogeneous collections of elements from some fixed domain. This is in contrast to other data-structures, like lists, which can only be accessed with recursive operators. The popularity of arrays stems from the fact that they can be used to model many abstractions useful in applications such as programming 
\cite{daca_array_2016, wang_solver_2023}, databases \cite{gianola_verification_2023, ding_proving_2023}, model checkers \cite{ghilardi_mcmt_2010}, memory models \cite{conchon_parameterized_2020}  or quantum circuits \cite{chen_theory_2023}.

Several theories of arrays in the literature express essentially the same concepts under different syntactic appearances \cite{stump_decision_2001, de_moura_generalized_2009, gleissenthall_cardinalities_2016, alberti_cardinality_2017}. As a consequence, a systematic classification of these theories is becoming increasingly difficult. This results in duplicated engineering efforts. It has been argued \cite[Lecture~19]{meseguer_lecture_2017} that some of these redundancies could be avoided by adopting a semantic perspective on the study of SMT theories. 

Our results show that the semantic approach is fruitful in the area of decision procedures for theories of arrays. We demonstrated in \cite{raya_vmcai_2022, raya_algebraic_2024} how, by fixing a model of such theories, we are able to reconstruct and extend the celebrated combinatory array logic fragment \cite{bradley_whats_2006}. In this paper, we further show how these observations can be extended to support summation constraints. Our methodology is inspired in the model theory of power structures \cite{mostowski_direct_1952, feferman_first_1959}, which we adapt from the first-order to the quantifier-free setting, which is the one relevant for applications to SMT.

\section{First-order model theory}
\label{section:prelim}
We start reviewing some notions from first-order model theory.

A \textbf{first-order language} is one whose logical symbols are $\lnot, \land, \lor, \forall$ and $\exists$, whose terms are either variables, constants or function symbols applied to terms and whose formulas are either atomic (relation symbols applied to terms) or general (atomic formulas and inductively, from formulas $A,B$, we get new formulas $\lnot A, A \land B$ and $A \lor B$ and from a formula $A$ and a variable symbol $x$ we get the new formulas $\exists x. A$ and $\forall x. A$).

A variable in a formula is free if there is no occurrence of a quantifier binding the variable name on the path of the syntax tree of the formula reaching the occurrence of the variable. A formula without free variables is a \textbf{sentence}. A \textbf{first-order theory} is a set of sentences written in some first-order language. 

A \textbf{first-order structure} $\mathcal{A}$ over a first-order language $L$ is a tuple with four components: a set $A$ called the domain of A; a set of elements of $A$ corresponding to the constant symbols of $L$; for each positive integer $n$, a set of $n$-ary relations on
$A$ (i.e. subsets of $A^n$), each of which is named by one or more $n$-ary relation symbols of $L$ and for each positive integer $n$, a set of $n$-ary operations on $A$ (i.e maps from $A^n$ to $A$), each of
which is named by one or more $n$-ary function symbols of $L$. The mapping assigning each first-order symbol of $L$ to its corresponding interpretation in $\mathcal{A}$ is denoted  $\cdot^{\mathcal{A}}$. This function is extended to work on terms, i.e. the application of function symbols to constants, variables or other terms, by requiring that $(f(s_1,\ldots,s_n))^{\mathcal{A}} = f^{\mathcal{A}}(s_1^{\mathcal{A}}, \ldots,s_n^{\mathcal{A}})$.

Let $\phi$ be a sentence in a first-order language $L$ and let $\cdot^{\mathcal{A}}$ be an interpretation of the symbols of $L$ in the structure $\mathcal{A}$. The \textbf{sentence $\phi$ is satisfied in the structure $\mathcal{A}$}, written $\mathcal{A} \models \phi$, if the following conditions apply.

\begin{itemize}
    \item[-] If $\phi$ is the atomic sentence $R(s_1,\ldots,s_n)$ where $s_1,\ldots,s_n$ are terms of $L$ then $\mathcal{A} \models \phi$ if and only if $(s_1^{\mathcal{A}},\ldots,s_n^{\mathcal{A}}) \in R^{\mathcal{A}}$.
    \item[-] $\mathcal{A} \models \lnot \phi$ if and only if it is not true that $\mathcal{A} \models \phi$. 
    \item[-] $\mathcal{A} \models \phi_1 \land \phi_2$ if and only if $\mathcal{A} \models \phi_1$ and $\mathcal{A} \models \phi_2$.
    \item[-] $\mathcal{A} \models \phi_1 \lor \phi_2$ if and only if $\mathcal{A} \models \phi_1$ or $\mathcal{A} \models \phi_2$. 
    \item[-] If $\phi$ is the sentence $\forall y. \psi(y)$ then $\mathcal{A} \models \phi$ if and only if for all elements $b$ of $A$, $\mathcal{A} \models \psi(b)$. 
    \item[-] If $\phi$ is the sentence $\exists y. \psi(y)$ then $\mathcal{A} \models \phi$ if and only if there is at least one element $b$ of $A$ such that $\mathcal{A} \models \psi(b)$. 
\end{itemize}

Let $Ax$ be a set of first-order sentences. We define the relation \textbf{$Ax \models \phi$} which holds if and only if for every structure $\mathcal{A}$, if $\mathcal{A} \models ax$ for each sentence $ax \in Ax$ then $\mathcal{A} \models \phi$. 

The \textbf{axiomatic theory} defined by a set of axioms $Ax$ is $Th(Ax) = \{ \phi | Ax \models \phi \}$. 

The \textbf{semantic theory} of a structure $\mathcal{A}$ is the set $Th(\mathcal{A}) := \{ \phi \, | \, \mathcal{A} \models \phi \}$.

When studying the sets $Th(\mathcal{A})$ and $Th(Ax)$ we may assume the sentences are in prenex normal form. A \textbf{prenex normal form} of a first-order formula $F$ is a first-order formula consisting of a string of quantifiers (called the prefix of the formula) followed by a quantifier-free formula (known as the matrix of the formula) which is equivalent to $F$. It is well-known that there is a polynomial time algorithm transforming sentences of a first-order theory into to equivalent sentences in prenex normal form.

The \textbf{existential fragment of the first-order theory $T$}, denoted $Th_{\exists^*}(T)$,is the subset of sentences in $T$ whose prefix in prenex normal form is purely existential. We write $Th_{\exists^*}(Ax)$ if $T$ is axiomatically specified and $Th_{\exists^*}(\mathcal{A})$ if $T$ is semantically specified.

\section{Array theories}
\label{section:arrays}
The theory of arrays $T_A$ is defined as a first-order theory with three sorts: $A$ for arrays, $I$ for indices and $E$ for elements of arrays. It has one ``read" function symbol $\cdot [ \cdot ]: A \times I \to E$, one ``write" function symbol $\cdot \langle \cdot \triangleleft \cdot \rangle: A \times I \times E \to A$ and includes the equality relation symbol $\cdot = \cdot$ for indices and elements. The theory is described axiomatically as the sets of sentences satisfying axioms $Ax$ of the following form \cite{bradley_calculus_2007}. $=$ is axiomatised as a 
reflexive, symmetric and transitive relation. Array read is assumed to be a congruence relation, i.e. $\forall a,i,j. i = j \to a[i] = a[j]$. Finally, there are axioms relating the read and write operations $\forall a,v,i,j. i = j \to a \langle i \triangleleft v\rangle[j] = v$ and $\forall a,v,i,j. i \neq j \to a \langle i \triangleleft v\rangle[j] = a[j]$.

The quantifier-free fragment of $T_A$ is the set of formulas that can be written without any use of quantifiers. Our goal is to decide which quantifier-free formulas are satisfiable. The satisfiable quantifier-free formulas correspond precisely to set of formulas in $Th_{\exists^*}(Ax)$.
\begin{proposition}
The existential closure of the satisfiable formulas in the quantifier-free fragment of $T_A$ is the set $Th_{\exists^*}(Ax)$. Conversely, if we drop the existential prefixes in $Th_{\exists^*}(Ax)$, we obtain the satisfiable formulas of the quantifier-free fragment of $T_A$.
\end{proposition}
\begin{proof}
The existential closure of a satisfiable formula in the quantifier-free fragment of $T_A$ is, by definition, in $Th_{\exists^*}(Ax)$. A formula in $Th_{\exists^*}(Ax)$ is true by definition. Converting it to prenex normal form and dropping the existential quantifier prefix leaves a formula of the quantifier-free fragment of $T_A$.
\end{proof}

Many works, starting with \cite{stump_decision_2001}, consider an extension of the theory $T_A$ with axioms of the form $R(a_1,\ldots,a_n) \leftrightarrow \forall i. R(a_1(i),\ldots,a_n(i))$ which says that some relation holds on a tuple of array variables $a_1,\ldots,a_n$ if and only it holds at each component. One example is the extensionality axiom $\forall a,b. a = b \leftrightarrow (\forall i. a[i] = b[i])$. In \cite{raya_vmcai_2022}, we observed that several fragments extending combinatory array logic \cite{de_moura_generalized_2009} can be described semantically as the theory of a power structure \cite{mostowski_direct_1952}. More precisely, we showed the following results. 

\begin{definition}
The generalised power $\mathcal{P}(\mathcal{A},I)$ of the combinatory array logic fragment is a structure whose carrier set is the set $M^I$ of functions from the index set $I$ to the carrier set of the structure of the array elements $M$ and whose relations are interpreted as sets of the form 
\[
\{  (a_1,\ldots,a_n) \in M^I | \Phi(S_1,\ldots,S_k) \}
\]
where $\Phi$ is a Boolean algebra expression over $\mathcal{P}(I)$ using the symbols $\subseteq$, $\cup$, $\cap$ or $\cdot^c$ and each set variable $S$ is interpreted as $S = \{ i \in I | \theta(a_1(i),\ldots,a_n(i)) \}$ where $\theta$ is a formula in the theory of the elements.
\end{definition}

\begin{theorem}
The quantifier-free formulas of combinatory array logic can be encoded in polynomial time as sentences in the theory of the generalised power $\mathcal{P}(\mathcal{A},I)$ in a way that preserves satisfiability of the formulas. 
\end{theorem}

\begin{theorem}
The theory $Th_{\exists^*}(\mathcal{P}(\mathcal{A},I))$ can be decided in $\NP$ even when the algebra of indices $\mathcal{P}(I)$ includes a cardinality operator and the language includes linear arithmetic constraints on the cardinality constraints. 
\end{theorem}

Our goal in this note is to generalise this result to summation constraints over the array (function symbols) variables. Interestingly, to preserve decidability one has to disallow constants in the element theory specifications $\theta$. 

\section{Decision Procedure}
\label{section:algebraic}
Our first step is defining the input language to be decided.

\begin{definition} \label{def:sums}
The theory of generalised powers with sums consists of formulae of the form
\begin{align} \label{eq:fragment}
\begin{split}
F(S_1,\ldots,S_k, \overline{\sigma}) \land \bigwedge_{i = 1}^k S_i = \{ \i \in I | \varphi_i(\overline{c}(\i)) \} \land \overline{\sigma} = \sum \mset{ \overline{c}(\i) ; \varphi_0(\overline{c}(\i))  }
\end{split}
 \end{align}
where $F$ is a formula from Boolean algebra of sets, $\varphi_0,\ldots,\varphi_k$ are formulae in the existential fragment of Presburger arithmetic and $\overline{c}$ is a tuple of arrays of natural numbers. We will refer to the first conjunct of this formula as the Boolean algebra term, to the second conjunct as the set interpretations and to the third conjunct as the multiset interpretations.
\end{definition}

There are some differences between Definition~\ref{def:sums} and \cite[Definition~2.1]{piskac_thesis_2011}.   \cite[Definition~2.1]{piskac_thesis_2011} has a quantifier-free Presburger arithmetic formula instead of the Boolean algebra term $F$. Second, the term $\forall e. F$ corresponds to our set interpretations. Third, the term $(u_1,\ldots,u_n) = \sum_{e \in E} (t_1,\ldots,t_n)$ corresponds to our multiset interpretation. It should be noted that the indices in our setting range over the natural numbers and not over a finite set $E$ as in \cite{piskac_thesis_2011}. 

An important observation is that the definition does not allow free variables to be shared between the three conjuncts. In fact, if we allowed such shared constants, the resulting fragment would have an undecidable satisfiability problem. 

\begin{corollary}
The satisfiability of formulas of the form
\begin{align} \label{eq:undecidable}
F(S_1,\ldots,S_k, \overline{\sigma}, \overline{f}) \land \bigwedge_{i = 1}^k S_i = \{ \i \in I | \varphi_i(\overline{c}(\i), \overline{f}) \} \land \overline{\sigma} = \sum \mset{ \overline{c}(\i) ; \varphi_0(\overline{c}(\i), \overline{f})  }
 \end{align}
 is undecidable.
\end{corollary}
\begin{proof}
By reduction from Hilbert's tenth problem \cite{matiyasevich_hilberts_1993}. One can encode in this theory the addition of two natural numbers using the formula $F$ which is in Boolean algebra of sets with cardinalities and thus includes quantifier-free Presburger arithmetic. Multiplication $z = x y$ can be encoded by imposing the array $\overline{c}$ to be equal to the constant $x$ in each position, have length $y$ and sum up to $z$. 
\end{proof}

Let us now describe the main steps of the decision procedure for the theory in Definition~\ref{def:sums}. 

\textbf{Elimination of terms in Boolean algebra with Cardinalities.} To eliminate these constraints, we introduce $k$ array variables $c_1,\ldots,c_k$ and we rewrite the Boolean algebra expressions and cardinality constraints in terms of set interpretations and summation constraints. See the appendix for further details.

As a result of this phase, we obtain a formula of the form:
\begin{align} \label{eq:sum}
\begin{split}
\psi(\overline{\sigma}) \land \bigwedge_{i = 1}^k I = \{ \i \in I | \varphi_i(\overline{c}(\i)) \} \land \overline{\sigma} = \sum \mset{ \overline{c}(\i) ; \varphi_0(\overline{c}(\i))  }
\end{split}
 \end{align}
where $\psi$ is a quantifier-free Presburger arithmetic formula and all the Boolean algebra and cardinality constraints has been translated into set interpretations and summation constraints.

\textbf{Elimination of the Set Interpretations.} The next step in the decision procedure is to eliminate the set interpretation term. However, in the form of Formula~\ref{eq:sum}, this is particularly simple. Formula~\ref{eq:sum} is equivalent to: 
\begin{align} \label{eq:interremoved}
\begin{split}
\psi(\overline{\sigma}) \land \overline{\sigma} = \sum \mset{ \overline{c}(\i) ; \bigwedge_{i = 0}^k \varphi_i(\overline{c}(\i))  }
\end{split}
\end{align}
It thus remains to remove the summation operator.

\textbf{Elimination of the Summation Operator.} The next step is to rewrite sums to a star operator introduced in \cite{piskac_linear_2008}. Given a set $A$, the set $A^*$ is defined as: 
\[
A^* = \set{ u ; \exists N \ge 0, x_1,\ldots, x_N \in A. u = \sum_{i = 1}^N x_i }
\]

\begin{proposition}[Multiset elimination] \label{lem:msetelim}
The formula 
\begin{equation}
\exists \overline{\sigma}, \overline{c}. \psi(\overline{\sigma}) \land \overline{\sigma} = \sum\limits_{n \in \N} \mset{ \overline{c}(n) ; \varphi(\overline{c}(n)) }
\end{equation}
and the formula
\begin{equation} \label{eq:star}
\exists \overline{\sigma}. \psi(\overline{\sigma}) \land \overline{\sigma} \in \set{\overline{k} ; \varphi(\overline{k})}^* 
\end{equation}
are equivalent.
\end{proposition}

The argument needs to be adapted from Theorem~2.4 of \cite{piskac_decision_2011} since both our index and element set are infinite. The details are given in the appendix.

The next step is to eliminate the star operator introduced in Proposition~\ref{lem:msetelim}. To do so, one could use \cite[Theorem~2.23]{piskac_decision_2011} which shows that if Formula~\ref{eq:star} is satisfiable then it also has a solution that can be written with a polynomial number of bits. We adapt this result to the case where we consider explicit integer exponents in the sets. That is we consider given a set $A$ and an integer $m \in \N$, the set $A^m$ defined as $A^m = \set{ u ; x_1,\ldots, x_m \in A. u = \sum_{i = 1}^m x_i }$. The reason to do this is that when mixing summation and other kinds constraints such as in \cite{raya_algebraic_2024}, we need to \textit{synchronise} the cardinality constraints of the combined theory with the cardinality constraints arising from the number of addends used in the sums. 

\begin{definition}
LIA with sum cardinalities, denoted $LIA^{card}$, is the theory consisting of formulas of the form $F_{0} \wedge \bigwedge_{i = 1}^n u \in \{x \mid F_i(x)\}^{x_i}$ where $F_{0}$ and $F$ are quantifier-free Presburger arithmetic formulae. 
\end{definition}

\begin{proposition} \label{prop:minimal}
$LIA^{card}$ is in $\NP$.
\end{proposition}

A detailed proof is given in the appendix.

\section{Conclusion}
\label{section:conclusion}
Despite the numerous works that are dedicated to the theory of arrays and its variations, it remains a challenge to provide a comprehensive classification of array theories according to the computational complexity of their satisfiability problem and their expressive power. This paper shows that even classical theories such as the combinatory array logic fragment can be optimised with respect to both metrics. An interesting extension that we leave open is to support combinatory array logic with sums and different element sorts.

\appendix
\section*{Appendix}
\addcontentsline{toc}{section}{Appendices}
\renewcommand{\thesubsection}{\Alph{subsection}}
\subsection{Elimination of Boolean algebra with Cardinalities terms}

\begin{itemize}
    \item For every newly introduced array variable $c_j$, we introduce the set interpretation:
    \[
    I = \{ \i \in I | c_j(\i) = 0 \lor c_j(\i) = 1 \} 
    \]
    this constraint is added to the set interpretation term.
    \item For every newly introduced array variable $c_j$, we rewrite the set variable $S_j$ into the following set interpretation:
    \[
    S_j := \{ i \in I | c_j(\i) = 1 \}
    \]
    which says that the indices in $S_j$ corresponds to the positions where $c_j$ is equal to one.
    \item We then substitute each Boolean algebra expression appearing in the Boolean algebra with cardinalities term of Formula~\ref{eq:fragment}, by repeatedly applying the following rewrite rules:
    \begin{align*}
    S_j^c &:= \{ \i \in I | c_j(\i) = 0 \} \\
    S_j \cup S_k &:= \{ \i \in I | c_j(\i) = 1 \lor c_k(\i) = 1 \} \\
    S_j \cap S_k &:= \{ \i \in I | c_j(\i) = 1 \land c_k(\i) = 1 \} 
    \end{align*}

    \item By the above rewriting process, each cardinality constraint $|S| = k$ is rewritten as 
    \[
    |\{ \i \in I | \varphi(\overline{c}(\i)) \}| = k
    \]
    where the variable $\overline{c}$ lists the newly introduced array variables $c_1,\ldots,c_k$. We then introduce a new array variable $x = ite(\varphi(\overline{c}), 1, 0)$, that is, $x(\i)$ is equal to one if $\varphi(\overline{c}(\i))$ holds and it is equal to $0$ otherwise. This can be encoded with the set interpretation
    \[
    I = \{ \i \in I | x(\i) = ite(\varphi(\overline{c}(\i)),1,0) \}
    \]
    One then rewrites the expression $|S| = k$ into $k = \sum_{i \in I} x(\i)$.\footnote{It is remarkable than being closed under the if-then-else operator (ite) also appears in modern presentations of Feferman-Vaught theorem, see for instance \cite[Theorem~9.6.2]{hodges_model_1993}, as the property of "being closed under gluing over a Boolean algebra". Such theoretical generalisations are also used in practice, see for instance \cite{schmid_generalized_2022}.}
\end{itemize}

\subsection{Elimination of the multiset comprehension}

\begin{proof}
\noindent $\Rightarrow)$ If (1) is satisfied, there are $\overline{\sigma}, \overline{c}$ such that $
\psi(\overline{\sigma}) \land \overline{\sigma} = \sum\limits_{n \in \N} \mset{ \overline{c}(n) ; \varphi(\overline{c}(n)) }
$. We claim that the same $\overline{\sigma}$ satisfies $\psi(\overline{\sigma}) \land \bigwedge_{i = 1}^k \overline{\sigma} \in \set{\overline{k} ; \varphi(\overline{k})}^*$. By hypothesis, $\psi(\overline{\sigma})$ is true. Moreover, $\overline{\sigma} = \mset{ \overline{c}(n) ; \varphi(\overline{c}(n)) }$ and either

\begin{itemize}
\item $\mset{ \overline{c}(n) ; \varphi(\overline{c}(n)) }$ is finite in which case $\overline{\sigma} \in \set{\overline{k} ; \varphi(\overline{k})}^*$.
\item or $\mset{ \overline{c}(n) ; \varphi(\overline{c}(n)) }$ is infinite, in which case $\overline{c}(n)$ is equal to $\overline{0}$ in all but a finite set of indices $I$ since by hypothesis the sum $\overline{\sigma}$ is finite. Then $\overline{\sigma} = \sum\limits_{n \in I} \overline{c}(n)$ and $\overline{\sigma} \in \set{\overline{k} ; \varphi(\overline{k})}^*$.
\end{itemize}

\noindent $\Leftarrow)$ If (2) is satisfied, then there is $\overline{\sigma}$ such that $\psi(\overline{\sigma}) \land \overline{\sigma} \in \set{\overline{k} ; \varphi(\overline{k})}^*$. It follows that there is a finite list of elements $\overline{k}_i$ such that $\overline{\sigma} = \sum_{i = 1}^p \overline{k}_i$. We define
\[
\overline{c}(n) = 
\begin{cases}
\overline{k}_n & \text{if } 1 \le n \le p \\
0 & \text{otherwise}
\end{cases}
\]
It is immediate that $\overline{\sigma}$ and $\overline{c}$ satisfy $\psi(\overline{\sigma}) \land \overline{\sigma} = \sum\limits_{n \in \N} \mset{ \overline{c}(n) ; \varphi(\overline{c}(n)) }$.
\end{proof}

\subsection{$\NP$ membership of $\lcard$}

\begin{proof}
Let $V_{PA}$ be a polynomial time verifier for $\lcard$. Figure~\ref{fig:cardverif} gives a verifier $V$ for $\lcard$. We show that $x \in \lcard$ if and only if there exists a polynomial-size certificate $w$ such that $V$ accepts $\langle x, w \rangle$.

\vspace{.5em}

\begin{figure*}[ht!]
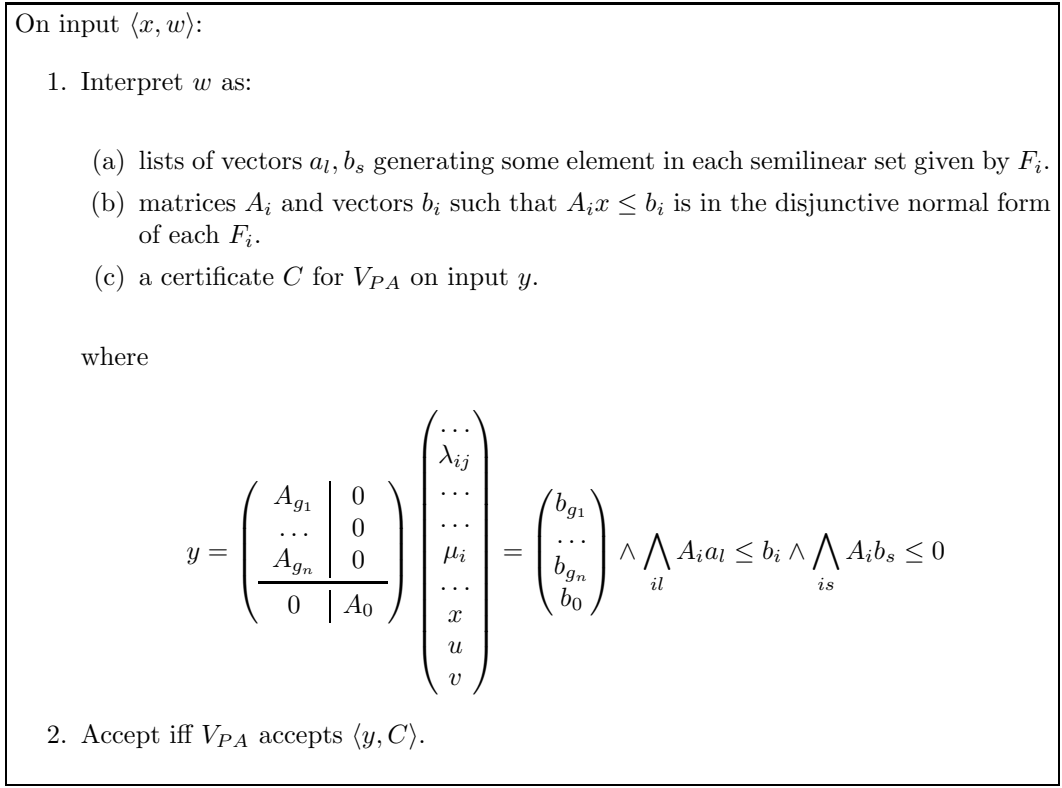

\fbox{\parbox{.95\textwidth}{
On input $\langle x, w \rangle$:

\begin{enumerate}
\setlength\itemsep{.5em}

\item Interpret $w$ as:
\vspace{1em}
\begin{enumerate}
    \item lists of vectors $a_l, b_s$ generating some element in each semilinear set given by $F_i$.
    \item matrices $A_i$ and vectors $b_i$ such that $A_ix \le b_i$ is in the disjunctive normal form of each $F_i$.
    \item a certificate $C$ for $V_{PA}$ on input $y$.
\end{enumerate}

\vspace{1em}

\noindent where 

\[ y = 
\left( \begin{array}{c|c}
   A_{g_1} & 0 \\
   \ldots  & 0 \\
   A_{g_n} & 0 \\
   \midrule
   0 & A_0 \\
\end{array}\right)
\begin{pmatrix}
\ldots \\
\lambda_{ij} \\
\ldots \\
\ldots \\
\mu_i \\
\ldots \\
x \\
u \\
v
\end{pmatrix} =
\begin{pmatrix}
b_{g_1} \\
\ldots  \\
b_{g_n} \\
b_0
\end{pmatrix}
\land 
\bigwedge_{il} A_i a_l \le b_i 
\land 
\bigwedge_{is} A_i b_s \le 0
\]

\item Accept iff $V_{PA}$ accepts $\langle y,C  \rangle$.
\end{enumerate}
}}
\caption{Verifier for $LIA^{card}$.} 
\label{fig:cardverif}
\end{figure*}

\noindent $\Rightarrow)$ If $x = F_{0} \wedge \bigwedge_{i = 1}^n u \in \{x \mid F_i(x)\}^{x_i} \in \lcard$ then we show that there is a solution that uses a polynomial number of bits in the size of $F_0$ and $F_1$.

We convert each formula $F_i$ to semilinear normal form \cite[Theorem~2.13]{piskac_decision_2011}.

\begin{lemma} \label{lem:semnf}
Let $F$ be a linear arithmetic formula of size $s$. Then there exist numbers $m, q_1, \ldots, q_m \in \mathbb{N}$ and vectors $a_i, b_{ij} \in \mathbb{N}^n$ for $1 \le j \le q_i, 1 \le i \le m$ with $\|a_i\|_1, \|b_{ij}\|_1 \le 2^{p(s)}$ with $p$  polynomial such that $F(x)$ is equivalent to the formula:
\begin{equation} \label{eq:semnf}
\exists \alpha_{11}, \ldots, \alpha_{mq_m}. \bigvee_{i = 1}^m \Big(x = a_i + \sum_{j =  1}^{q_i} \alpha_{ij} b_{ij}\Big)
\end{equation}
\end{lemma}

Next, we eliminate the star operator, as in \cite[Proposition~2]{lugiez_multitrees_2002} and \cite[Theorem~2.14]{piskac_decision_2011}.

\begin{lemma} \label{lem:starelim}
Let $F$ be a quantifier-free linear integer arithmetic formula whose semilinear normal form is formula \ref{eq:semnf}.
Then $u \in \set{ y ; F(y) }^x$ is equivalent to
\[
\exists \overline{\mu}, \overline{\lambda}. u = \sum_{i = 1}^q \Bigg(\mu_i a_i + \sum_{j = 1}^{q_i} \lambda_{ij} b_{ij}\Bigg) \land \bigwedge_{i = 1}^q \Bigg(\mu_i = 0 \implies \sum_{j = 1}^{q_i} \lambda_{ij} = 0\Bigg) \land x = \sum_{i = 1}^q \mu_i
\]
\end{lemma}

We express the resulting vector $u$ with polynomially many generators \cite[Theorem~2.20]{piskac_decision_2011}.

\begin{lemma}[Polynomially many generators for sums] \label{lem:polsum}
Let $F$ be a quantifier-free linear integer arithmetic formula of size $s$ whose semilinear normal form is formula \ref{eq:semnf}. Then $u \in \set{ y ; F(y) }^x$ is equisatisfiable with
\[
\exists \lambda_{ij}, \mu_i.u = \sum_{i \in I_0} \Big( a_i + \sum_{(i,j) \in J} \lambda_{ij} b_{ij} \Big) + \sum_{i \in I_1} \mu_i a_i \land x = |I_0| + \sum_{i \in I_1} \mu_i
\]
for some $I_0, I_1 \subseteq \set{ 1, \ldots, q}$, $J \subseteq \cup_{i = 1}^q \set{ (i,1), \ldots, (i,q_i) }$, $|I_0| \le |J| \le q(s)$, $|I_1| \le q(s)$ and $q$ is a polynomial.
\end{lemma}

We write the equation 
\[
u = \sum_{i \in I_0} \Big( a_i + \sum_{(i,j) \in J} \lambda_{ij} b_{ij} \Big) + \sum_{i \in I_1} \mu_i a_i \land x = |I_0| + \sum_{i \in I_1} \mu_i
\]
as a system $A_g x_g = b_g$ of the form


\[
\left( 
\begin{array}{ccccccccccc}
\ldots & a_i    & \ldots & \ldots & b_{ij} & \ldots & \ldots & a_i    & \ldots & 0 & - I_k  \\
0      & \ldots & \ldots & \ldots & \ldots & 0      & 1      & \ldots & 1      & -1 & 0    \\
\end{array}
\right)
\begin{pmatrix}
\ldots \\
\lambda_{ij} \\
\ldots \\
\ldots \\
\mu_i \\ 
\ldots \\
x \\
u 
\end{pmatrix} =
\begin{pmatrix}
 \ldots \\
 - a_i \\
 \ldots \\
0 \\
\ldots \\
0 \\
-|I_0|
\end{pmatrix}
\]
Since $u$ is also a solution of $F_0$ it will further satisfy some system $A_0 w = b_0$ corresponding to one of the terms of the disjunctive normal form of $F_0$ and where we can assume that the unknown $u$ appears in the first rows of $w = (u, v)$.  As a result, we obtain a combined system.
\[
\left( \begin{array}{c|c}
   A_{g_1} & 0      \\
   \ldots  & \ldots \\
   A_{g_n} & 0      \\
   \midrule
   0 & A_0 \\
\end{array}\right)
\begin{pmatrix}
\ldots \\
\lambda_{ij} \\
\ldots \\
\ldots \\
\mu_i \\
\ldots \\
x \\
u \\
v
\end{pmatrix} =
\begin{pmatrix}
b_{g_1} \\
\ldots  \\
b_{g_n} \\
b_0
\end{pmatrix}
\]
This system has a polynomial number of rows, columns and uses polynomially many bits for its maximum absolute value.  Thus, it is guaranteed to have a solution that uses polynomially many bits by  the following theorem from \cite[page 767]{papadimitriou_complexity_1981}.

\begin{lemma}
Let A be an $m \times n$ integer matrix and $b$ a $m$-vector, both with entries from $[-a . . a]$. Then the system $A x=b$ has a solution in $\mathbb{N}^{n}$ if and only if it has a solution in $[0 . . M]^{n}$ where $M=n(m a)^{2 m+1}$.
\end{lemma}

Moreover, since  all the generators chosen for $F_i$ lie in the same linear subset they satisfy $A_i a_l \le b_i$ and $A_i b_s \le 0$ for some system $A_ix \le b_i$ in the disjunctive normal form of each $F_i$. 

The resulting formula has polynomial-size in the size of $F_0$ and $F_1$, thus there exists a polynomial-size certificate $C$ such that $V_{PA}$ accepts $\langle y, C \rangle$. It follows that $V_{\lcard}$ accepts $\langle x, \langle \{a_l\}, \{b_s\}, \{A_i\}, \{b_i\}, C \rangle \rangle$.

\vspace{.5em}

\noindent $\Leftarrow)$ If $V_{\lcard}$ accepts $\langle x, w \rangle$ then there exists a polynomial-size certificate $C$ such that $V_{PA}$ accepts $\langle y, C \rangle$ and thus $y$ is satisfiable. This means that for some vectors $a_l$ satisfying $F_i$ and some vectors $b_s$ satisfying the homogeneous part of $F_i$ it holds that
\[
u = \sum_{l \in I_0} \Big( a_l + \sum_{(l,s) \in J} \lambda_{ls} b_{ls} \Big) + \sum_{l \in I_1} \mu_l a_l \land x_i = |I_0| + \sum_{l \in I_1} \mu_l
\]
\noindent and furthermore $(u,v)$ satisfies $F_0$. It follows that $F_{0} \wedge \bigwedge_{i = 1}^n u \in \{x \mid F_i(x)\}^{x_i}$ is satisfied by such $u$ and $v$. Thus, $x \in \lcard$.
\end{proof}

\end{document}